\documentclass[submission,copyright,creativecommons]{eptcs}
\providecommand{\event}{FROM 2026} % Name of the event you are submitting to

\usepackage{iftex}

\ifpdf
  \usepackage{underscore}         % Only needed if you use pdflatex.
  \usepackage[T1]{fontenc}        % Recommended with pdflatex
\else
  \usepackage{breakurl}           % Not needed if you use pdflatex only.
\fi

\usepackage{amsmath}
\usepackage{amsfonts}
\usepackage{amssymb}
\usepackage{amsthm}
\usepackage{graphicx}
\usepackage{hyperref}

\newcommand{\ML}{\mathbb{M}\mathbb{L}}
\newcommand{\iML}{\operatorname{i}\mathbb{M}\mathbb{L}}
\newcommand{\ite}[3]{\operatorname{if}\;#1\;\operatorname{then}\;#2\;\operatorname{else}\;#3}
\newcommand{\abs}[1]{\lvert#1\rvert}
\renewcommand{\phi}{\varphi}

\newcommand{\himp}{\Rightarrow}
\let\oldbigsqcup\bigsqcup
\renewcommand{\bigsqcup}{\mathop{\scalebox{1.2}{\reflectbox{\rotatebox[origin=c]{180}{$\bigsqcap$}}}}\displaylimits}
\newcommand{\bigsqcap}{\mathop{\scalebox{1.2}{\reflectbox{\rotatebox[origin=c]{180}{$\oldbigsqcup$}}}}\displaylimits}
\renewcommand{\operatorname}[1]{\mathsf{#1}} 

\newtheorem{theorem}{Theorem}
\newtheorem{lemma}{Lemma}

\title{Complete Heyting Algebra Semantics for an Intuitionistic Version of Matching Logic \\ \large{(Extended Abstract)}}
\author{Hora\c{t}iu Cheval
\institute{Research Center for Logic, Optimization and Security (LOS), University of Bucharest, Romania}
\institute{Institute for Logic and Data Science, Bucharest, Romania}
\email{horatiu.cheval@unibuc.ro}}
\def\titlerunning{Intuitionistic AML}
\def\authorrunning{H. Cheval}
\begin{document}
\maketitle

\begin{abstract}
  We present work in progress towards an intuitionistic version of Applicative Matching Logic. 
  We introduce a semantics based on complete Heyting algebras, and propose a proof system 
  which we prove to be sound relative to this semantics. 
\end{abstract}

\section{Introduction}

Matching Logic \cite{rosu2017matching, chen2019matching, chen2021explained} is a minimal yet expressive logical system, 
primarily designed for reasoning in a language-independent manner about program executions. 
It serves, moreover, as the logical foundation of the $\mathbb{K}$ program verification framework (\url{https://kframework.org}).
Several versions of Matching Logic have been introduced throughout the years, aiming at increasing the expressivity of the system 
while simultaneously simplifying its syntax and proof system. The version we concern ourselves with in this paper
is the most recent one, also known as \emph{Applicative} Matching Logic, and we will refer to it simply as Matching Logic or $\ML$.
A comprehensive overview is provided in \cite{chen2023matching}.

The goal and main contribution of this paper is to propose an intuitionistic version of $\ML$. 
Let us first briefly recall the syntax and semantics of $\ML$, 
as a reference to compare our intuitionistic system against. 
A distinctive feature of $\ML$ is that its syntax is composed of only one kind of objects,
called \emph{patterns}, which play the role of both terms and formulas from first-order logic. 
Given a signature $\mathbb{S} = (EV, SV, \Sigma)$ consisting of (disjoint)
sets of \emph{element variables}, \emph{set variables}, and \emph{symbols}, the set of patterns over $\mathbb{S}$, $\operatorname{Patt}_{\mathbb{S}}$ is generated by 
\begin{align*}
  \phi, \psi ::= \ 
    &x \in EV & (\text{element variable}) \\
    &\mid X \in SV &(\text{set variable}) \\ 
    &\mid \sigma \in \Sigma & (\text{symbol})\\ 
    &\mid \phi \psi & (\text{application}) \\ 
    &\mid \phi \to \psi & (\text{implication}) \\ 
    &\mid \exists x. \phi & (\text{existential quantifier}) \\ 
    & \mid \mu X. \phi & (\text{least fixed point})
\end{align*}
When constructing $\mu X. \phi$, it is required moreover that $\phi$ is positive for $X$,
meaning that $X$ does not occur in $\phi$ on the left hand-side of an odd number of nested implications.
Patterns denote subsets of a fixed domain $M$, with symbols and application not having an 
a priori interpretation, instead receiving one from the interpreting structure. 
Concretely, a model is a tuple $\mathcal{M} = \left(M, (M_\sigma)_{\sigma \in \Sigma}, \_ \cdot \_\right)$,
where, for any $\sigma \in \Sigma$, $M_\sigma \in \mathcal{P}(M)$ is the interpretation of the symbol $\sigma$,
and $\_ \cdot \_ : M \times M \to \mathcal{P}(M)$ interprets application.
A valuation in $M$ is a mapping $\rho : EV \cup\, SV \to M \cup \mathcal{P}(M)$ such that 
$\rho(x) \in M$, for any $x \in EV$, and $\rho(X) \in \mathcal{P}(M)$ for any $X \in SV$. 
That is to say, element variables stand for single elements of the domain, while set variables stand for subsets.
Then, $\abs{\_}_\rho : \operatorname{Patt}_{\mathbb{S}} \to \mathcal{P}(M)$ 
is defined by 
\begin{equation*}
  \abs{x}_\rho = \{\rho(x)\}, \quad
  \abs{X}_\rho = \rho(X), \quad
  \abs{\sigma}_\rho = M_\sigma, \quad 
  \abs{\phi \to \psi}_\rho = (M \setminus \abs{\phi}_\rho) \cup \abs{\psi}_\rho, 
\end{equation*}
\begin{equation*}
  \abs{\phi \psi}_\rho = \bigcup_{a \in \abs{\phi}_\rho} \bigcup_{b \in \abs{\psi}_\rho} a \cdot b, \quad
  \abs{\exists x. \phi}_\rho = \bigcup_{a \in M} \abs{\phi}_{\rho[x \mapsto a]}, \quad
  \abs{\mu X. \phi}_\rho = \operatorname{lfp}\left(A \in \mathcal{P}(M) \mapsto \abs{\phi}_{\rho[X \mapsto A]}\right) 
\end{equation*}
where $\rho[x \mapsto a]$ denotes the updated valuation 
$\rho[x \mapsto a](y) = \ite{x = y}{a}{\rho(y)}$, and similarly for set variables.
The logical connectives denote their natural set-theoretic counterparts, 
while application is the union of all the pointwise applications between elements of $\phi$ and $\psi$.
The positivity condition in the formation of $\mu X. \phi$ ensures that the mapping $A \mapsto \abs{\phi}_{\rho[X \mapsto A]}$ is monotone,
and subsequently the existence of its least fixed point, by the Knaster-Tarski theorem.

Without writing it down here, let us conclude this section by mentioning that $\ML$ 
possesses a sound and conditionally complete proof system \cite{chen2021explained}.

\section{Intuitionistic $\ML$}

We now proceed to describe the intuitionistic system we consider, which we will refer to as $\iML$. 
In terms of syntax, the only difference is the addition of constructors which are no longer inter-derivable intuitionistically.
Same as in the classical version of $\ML$, a signature $\mathbb{S} = (EV, SV, \Sigma)$ 
contains disjoint sets of element variables, set variables, and symbols. 
Then, the set of patterns over the signature $\mathbb{S}$, $\operatorname{IPatt}_{\mathbb{S}}$ is given by 
\begin{equation*}
  \phi, \psi ::= 
    x \in EV 
    \mid X \in SV 
    \mid \sigma \in \Sigma 
    \mid \phi \psi
    \mid \phi \wedge \psi 
    \mid \phi \vee \psi 
    \mid \phi \to \psi 
    \mid \exists x. \phi 
    \mid \forall x. \phi 
    \mid \mu X. \phi 
    \mid \nu X. \phi 
\end{equation*}
For $\mu X. \phi$ and $\nu X. \phi$, we still require that $\phi$ be positive for $X$.
The constructors 
$\wedge, \vee, \forall$ and $\nu$ are added explicitly to the syntax, as it is common 
for intuitionistic systems.
We do not however include $\bot$ as a primitive, 
as is common in intuitionistic logic,
since instead we can define it by 
$\bot := \mu X. X$ (this may also be done for classical $\ML$).
Negation is then standardly defined as $\neg \phi := \phi \to \bot$.
We denote by $\phi[\psi/X]$ and $\phi[\psi/x]$ the capture-avoiding substitution 
of set variable $X$ (resp. element variable $x$) by the pattern $\psi$ within $\phi$.

\subsection{Semantics}

In classical $\ML$, patterns denote subsets of a fixed domain $M$,
best seen here as functions $M \to \{0, 1\}$.
The idea of our semantics for intuitionistic $\ML$ will be to replace 
$\{0, 1\}$ with a complete Heyting algebra $\mathbb{H}$, 
with patterns then denoting, instead of subsets, mappings $M \to \mathbb{H}$. 

A complete Heyting algebra $\mathbb{H}$ is a complete lattice satisfying the following distributivity law:
for any $x \in \mathbb{H}$ and $A \subseteq \mathbb{H}$, $x \sqcap \bigsqcup_{y \in A} y = \bigsqcup_{y \in A} (x \sqcap y)$. 
In particular, $\mathbb{H}$ admits a Heyting implication definable by $x \himp y = \bigsqcup \{z \mid z \sqcap x \leq y\}$.
The property we mostly rely on is that $x \himp y = \top$ if and only if $x \leq y$. 

We also recall the well-known Knaster-Tarski theorem, which will be used to interpret the fixed point operators.
\begin{theorem}[Knaster-Tarski]\label{thm:kt}
  Let $L$ be a complete lattice and $f : L \to L$ be a monotone function. 
  Then, $f$ has a least and a greatest fixed point, 
  given by $\operatorname{lfp}(f) = \bigsqcap \{ x \in L \mid f(x) \leq x \}$
  and $\operatorname{gfp}(f) = \bigsqcup \{ x \in L \mid x \leq f(x) \}$.
\end{theorem}

A model is again a tuple $\mathcal{M} = \left(\mathbb{H}, M, (M_\sigma)_{\sigma \in \Sigma}, \_ \cdot \_\right)$, where
\begin{itemize}
  \item $\mathbb{H}$ is a complete Heyting algebra,
  \item $M$ is a set,
  \item for any $\sigma \in \Sigma$, $M_\sigma : M \to \mathbb{H}$,
  \item $\_ \cdot \_ : M \times M \to \mathbb{H}^M$.
\end{itemize}
An $\mathcal{M}$-valuation is a pair of functions $\rho = (\rho_{EV} : EV \to M, \rho_{SV} : SV \to \mathbb{H}^M)$,
assigning values to element and set variables.
Given such a valuation, the interpretation of patterns 
$\abs{\_}_\rho : \operatorname{IPatt}_{\mathbb{S}} \to \mathbb{H}^M$ is defined recursively, as follows: 
\begin{align*}
  \abs{x}_\rho &= m \mapsto \ite{\rho_{EV}(x) = m}{\top}{\bot} \\
  \abs{X}_\rho &= \rho_{SV}(X) \\ 
  \abs{\sigma}_\rho &= M_\sigma \\ 
  \abs{\phi \psi}_\rho &= m \mapsto \bigsqcup_{a \in M} \bigsqcup_{b \in M} \abs{\phi}_\rho(a) \sqcap \abs{\psi}_\rho(b) \sqcap (a \cdot b)(m) \\
  % \abs{\bot}_\rho (m) &= \bot \\
  \abs{\phi \wedge \psi}_\rho &= m \mapsto \abs{\phi}_\rho(m) \sqcap \abs{\psi}_\rho(m) \\
  \abs{\phi \vee \psi}_\rho  &= m \mapsto \abs{\phi}_\rho(m) \sqcup \abs{\psi}_\rho(m) \\
  \abs{\phi \to \psi}_\rho &= m \mapsto \abs{\phi}_\rho(m) \himp \abs{\psi}_\rho(m) \\
  \abs{\exists x. \phi}_\rho  &= m \mapsto \bigsqcup_{a \in M} \abs{\phi}_{\rho[x \mapsto a]}(m) \\
  \abs{\forall x. \phi}_\rho &= m \mapsto \bigsqcap_{a \in M} \abs{\phi}_{\rho[x \mapsto a]}(m) \\
  \abs{\mu X. \phi}_\rho &= \operatorname{lfp}(A \in \mathbb{H}^M \mapsto \abs{\phi}_{\rho[X \mapsto A]}) \\
  \abs{\nu X. \phi}_\rho &= \operatorname{gfp}(A \in \mathbb{H}^M \mapsto \abs{\phi}_{\rho[X \mapsto A]})
\end{align*} 

The clauses for the logical connectives are easily observed to be the natural (pointwise) Heyting counterparts of the classical set-theoretic operations. 
The same motivation becomes apparent for application, if we unfold its classical semantics into 
$m \in \abs{\phi \psi}_\rho \leftrightarrow \exists a, b \in M \left(a \in \abs{\phi}_\rho \wedge b \in \abs{\psi}_{\rho} \wedge m \in a \cdot b\right)$.
As in the classical setting, the clauses for the $\mu$ and $\nu$ constructors rely on the syntactical positivity condition, 
as made precise by the following technical lemma, which we state without proof here. 
\begin{lemma}
  Let $\phi$ be a pattern and $X$ be a variable such that $\phi$ is positive for $X$. 
  Define $F : \mathbb{H}^M \to \mathbb{H}^M$ by $F(A) = \abs{\phi}_{\rho[X \mapsto A]}$,
  as used in the definition of the semantics. 
  Then, $F$ is monotone.
\end{lemma}
Thus, by the Knaster-Tarski Theorem, which still applies in our setting since $\mathbb{H}^M$ forms 
a complete lattice, $F$ has a least and a greatest fixed point.

A pattern $\phi$ is said to be satisfied under the $\mathcal{M}$-valuation $\rho$,
written $\mathcal{M}, \rho \models \phi$, if
for any $m \in M$, $\abs{\phi}_\rho(m) = \top$. 
A model $\mathcal{M}$ satisfies $\phi$ if 
$\mathcal{M}, \rho \models \phi$ for every $\mathcal{M}$-valuation $\rho$,
and it satisfies a set of patterns $\Gamma$ if it satisfies all patterns in $\Gamma$. 
We write $\mathcal{M} \models \phi$ and $\mathcal{M} \models \Gamma$ for this. 
Finally, for a set of patterns $\Gamma$, we say that $\Gamma \models \phi$ if 
for any model $\mathcal{M}$, $\mathcal{M} \models \Gamma$ implies $\mathcal{M} \models \phi$.

\subsection{Proof system}

We now present the proposed Hilbert-style proof system for $\iML$, which compares to 
the classical one in the following way.
First, we use a standard \cite{troelstra1973metamathematical} intuitionistic formulation for the propositional and first-order parts of the proof system.
Such an axiomatization was also recently studied in relation to $\ML$ in \cite{leustean2025matching}. 
The least fixed point rules from $\ML$ are split into corresponding least and greatest fixed point rules,
since, as far as we can tell, their classical inter-derivability needs double negation elimination. 
Finally, the proof rules related to application are left the same as in the classical setting,
since, as we will show, they are intuitionistically valid.

Before giving the full proof system, 
we need to introduce the notion of an application context which appears in the $(\operatorname{SINGLETON})$ axiom.
Intuitively, an application context represents a pattern comprised only of nested applications, 
starting from a hole which may be plugged by any pattern. Formally, they can be described by the following grammar:
\begin{equation*}
  C ::= \square \mid \phi C \mid C \phi, 
\end{equation*}
For a pattern $\psi$ and an application context $C$, 
we denote by $C[\psi]$ the pattern obtained by replacing the hole in $C$ by $\psi$,
defined by recursion on $C$:
\begin{align*}
  \square[\psi] &= \psi \\ 
  (\phi C)[\psi] &= \phi (C[\psi]) \\ 
  (C \phi)[\psi] &= (C[\psi]) \phi
\end{align*}
Let us point out that these definitions are identical to the classical case. 
The proof rules we consider are given in Figure~\ref{fig:proof-rules}.

\mbox{}\\[4mm]
The absence of the typical $(\operatorname{EXFALSO})$ rule asserting that 
$\bot \to \phi$ is due to the fact that, 
with our definition of $\bot$ as $\mu X. X$, it can be easily derived from 
$(\operatorname{KNASTER\text{-}TARSKI})$: \\[2mm]
\begin{tabular}{lll}
  (1) $\vdash \phi \to \phi$ & ($\operatorname{WEAKENING}$), ($\operatorname{CONTRACTION}$) and ($\operatorname{SYLLOGISM}$) \\ 
  (2) $\vdash X[\phi/X] \to \phi$ & (since $X[\phi/X] = \phi$) \\ 
  (3) $\vdash \mu X. X \to \phi$ & ($\operatorname{KNASTER\text{-}TARSKI}$): (2) \\  
  (4) $\vdash \bot \to \phi$ & by the definition of $\bot$ 
\end{tabular} 

We also do not include the $(\operatorname{PROPAGATION}_\bot)$ rule
encountered in some formulations of Matching Logic, 
since it is known that it can be derived in presence of $(\operatorname{FRAMING})$ and $(\operatorname{SINGLETON})$,
and the proof can still be carried out intuitionistically. 
See \cite[Proposition~3.3]{chen2023matching} for a proof.

\section{Soundness}
In this section, we sketch the soundness of the proof system from Figure \ref{fig:proof-rules} 
with respect to the semantics we introduced. 
We point out that the proof was also formalized in the Lean theorem prover,
mostly autonomously with interactive user feedback using Claude 4.6 and 4.8 Opus, running in Claude Code.\footnote{The development is available at \url{https://github.com/hcheval/IntuitionisticML}.}
This was not an autoformalization of a result fully written on paper,
instead happening in sync with developing the theory,
and the process helped us rule out some early unsound or otherwise degenerate versions of the system.
The following is a standard technical lemma that is used in the proof.
\begin{lemma}
  For any patterns $\phi$ and $\psi$, and any valuation $\rho$,
  \begin{align*}
    \abs{\phi[\psi/X]}_\rho = \abs{\phi}_{\rho[X \mapsto \abs{\psi}_\rho]}.
  \end{align*}
\end{lemma}

The following is a technical result that will be used in the soundness proof.
\begin{lemma}\label{lem:fill-le-iSup}
Let $C$ be an application context and $\psi$ be a pattern. Then, for any valuation $\rho$ and any $m \in M$, 
we have that 
$\abs{C[\psi]}_\rho(m) \leq \bigsqcup_{a \in M} \abs{\psi}_\rho(a)$.
\end{lemma}
\begin{proof}
  By induction on $C$. If $C = \square$, then $C[\psi] = \psi$ so the claim is trivial. 
  Suppose now that $C = \phi C'$. For any $b, c \in M$,
  \begin{align*}
    \abs{\phi}_\rho(b) \sqcap \abs{C'[\psi]}_\rho(c) \sqcap (b \cdot c)(m) 
    &\leq \abs{\phi}_\rho(b) \sqcap \bigsqcup_{a \in M} \abs{\psi}_\rho(a) \sqcap (b \cdot c)(m) \\
    &\leq \bigsqcup_{a \in M} \abs{\psi}_\rho(a) 
  \end{align*} 
  Therefore, we can conclude that 
  \begin{align*}
    \abs{C[\psi]}_\rho(m) = \abs{(\phi C')[\psi]}_\rho(m) 
    &= \bigsqcup_{a \in M} \bigsqcup_{b \in M} \abs{\phi}_\rho(a) \sqcap \abs{C'[\psi]}_\rho(b) \sqcap (a \cdot b)(m) \\
    &\leq \bigsqcup_{a \in M} \abs{\psi}_\rho(a)
  \end{align*}
  The case of $C = C'\phi$ is analogous.
\end{proof}

\begin{theorem}
  Let $\mathbb{S}$ be a signature 
  and $\Gamma$ be a set of patterns over $\mathbb{S}$.
  If $\Gamma \vdash \phi$, then $\Gamma \models \phi$.
\end{theorem}
\begin{proof}
  We argue by induction on the derivation. 
  For space constraints, we only sketch the more interesting cases. 
  For the propositional and first-order cases, the rules and semantics considered here are standard, 
  so we omit them. The $(\operatorname{EXISTENCE})$ rule
  can also be easily treated similarly to the classical case.
  We only consider the left case for the 
  $(\operatorname{FRAMING})$, $(\operatorname{PROPAGATION}_\vee)$ and $(\operatorname{PROPAGATION}_\exists)$
  rules as the other is handled analogously.
  The only somewhat unexpected case is that of the $(\operatorname{SINGLETON})$ rule,
  which requires a different strategy than in the classical setting.
  Let, throughout, $\mathcal{M}$ be a model and $\rho$ be an $\mathcal{M}$-valuation. 
  \begin{itemize}
    \item $(\operatorname{PROPAGATION}_\vee)$:  
    We need to show that 
    \begin{align*}
      &\bigsqcup_{a \in M} \bigsqcup_{b \in M} \left(\abs{\phi}_\rho(a) \sqcup \abs{\psi}_\rho(a)\right) \sqcap \abs{\chi}_\rho(b) \sqcap (a \cdot b)(m) \\
      &\quad\leq 
      \left(\bigsqcup_{a \in M} \bigsqcup_{b \in M} \abs{\phi}_\rho(a) \sqcap \abs{\chi}_\rho(b) \sqcap (a \cdot b)(m)\right) 
      \sqcup  
      \left(\bigsqcup_{a \in M} \bigsqcup_{b \in M} \abs{\psi}_\rho(a) \sqcap \abs{\chi}_\rho(b) \sqcap (a \cdot b)(m)\right)
    \end{align*}
    \begin{align*}
      &\bigsqcup_{a \in M} \bigsqcup_{b \in M} \left(\abs{\phi}_\rho(a) \sqcup \abs{\psi}_\rho(a)\right) \sqcap \abs{\chi}_\rho(b) \sqcap (a \cdot b)(m) \\
      &\quad= \bigsqcup_{a \in M} \bigsqcup_{b \in M} \left(\abs{\phi}_\rho(a) \sqcap \abs{\chi}_\rho(b) \sqcap(a \cdot b)(m) \right) 
      \sqcup \left(\abs{\psi}_\rho(a) \sqcap \abs{\chi}_\rho(b) \sqcap (a \cdot b)(m)\right) \\ 
      &\quad= \bigsqcup_{a \in M} \bigsqcup_{b \in M} \left(\abs{\phi}_\rho(a) \sqcap \abs{\chi}_\rho(b) \sqcap(a \cdot b)(m) \right) 
      \sqcup \bigsqcup_{a \in M} \bigsqcup_{b \in M} \left(\abs{\psi}_\rho(a) \sqcap \abs{\chi}_\rho(b) \sqcap (a \cdot b)(m)\right)
    \end{align*}
    which is exactly the desired right-hand side. 
    Note here that we actually got equality instead of just an upper bound, 
    suggesting that the reverse implication also holds. This is indeed the case, as in the classical $\ML$, 
    but the propagation rules suffice to be stated for this direction only, as the other implication can be obtained 
    via $(\operatorname{FRAMING})$.

    \item $(\operatorname{PROPAGATION}_\exists)$: 
    We have that
    \begin{align*}
      &\bigsqcup_{a \in M} \bigsqcup_{b \in M} \abs{\exists x. \phi}_\rho(a) \sqcap \abs{\chi}_\rho(b) \sqcap (a \cdot b)(m) \\ 
      &\quad = \bigsqcup_{a \in M} \bigsqcup_{b \in M} \left( \bigsqcup_{c \in M} \abs{\phi}_{\rho[x \mapsto c]}(a) \right) \sqcap \abs{\chi}_{\rho}(b) \sqcap (a \cdot b)(m) \\ 
      &\quad = \bigsqcup_{a \in M} \bigsqcup_{b \in M} \left( \bigsqcup_{c \in M} \abs{\phi}_{\rho[x \mapsto c]}(a) \sqcap \abs{\chi}_{\rho}(b) \sqcap (a \cdot b)(m) \right) & \text{(frame distributivity)} \\
      &\quad = \bigsqcup_{a \in M} \bigsqcup_{b \in M} \left( \bigsqcup_{c \in M} \abs{\phi}_{\rho[x \mapsto c]}(a) \sqcap \abs{\chi}_{\rho[x \mapsto c]}(b) \sqcap (a \cdot b)(m) \right) & \text{(since $x \notin FV(\chi)$)} \\
      &\quad = \bigsqcup_{c \in M} \bigsqcup_{a \in M} \bigsqcup_{b \in M} \abs{\phi}_{\rho[x \mapsto c]}(a) \sqcap \abs{\chi}_{\rho[x \mapsto c]}(b) \sqcap (a \cdot b)(m) \\
      &\quad = \bigsqcup_{c \in M} \abs{\phi \chi}_{\rho[x \mapsto c]}(m) = \abs{\exists x. \phi \chi}_\rho(m)
    \end{align*}

    \item $(\operatorname{FRAMING})$: 
    We have to prove that $\abs{\phi \chi \to \psi \chi}_\rho(m) = \top$,
    which is equivalent to showing that 
    \begin{equation*}
      \bigsqcup_{a \in M} \bigsqcup_{b \in M} \abs{\phi}_\rho(a) \sqcap \abs{\chi}_\rho(b) \sqcap (a \cdot b)(m) 
      \leq  
      \bigsqcup_{a \in M} \bigsqcup_{b \in M} \abs{\psi}_\rho(a) \sqcap \abs{\chi}_\rho(b) \sqcap (a \cdot b)(m),
    \end{equation*}
    which follows immediately from the Induction Hypothesis, which asserts that for all $m \in M$, 
    $\abs{\phi}_\rho(m) \leq \abs{\psi}_\rho(m)$.

    \item $(\operatorname{SINGLETON})$:
    We need to prove that $\abs{C_1[x \wedge \phi]}_\rho(m) \sqcap \abs{C_2[x \wedge \neg\phi]}_\rho(m) = \bot$.
    By the bound in Lemma~\ref{lem:fill-le-iSup}, it suffices to prove that 
    \begin{align*}
      \bigsqcup_{a \in M} \abs{x \wedge \phi}_\rho(a) 
        \sqcap 
      \bigsqcup_{a \in M} \abs{x \wedge \neg\phi}_\rho(a) = \bot
    \end{align*}
    By frame distributivity, it suffices to show that for any $a, b \in M$,
    \begin{align*}
      \abs{x \wedge \phi}_\rho(a) \sqcap \abs{x \wedge \neg\phi}_\rho(b) = \bot.
    \end{align*}
    By unfolding the definition we get 
    \begin{align*}
      \abs{x \wedge \phi}_\rho(a) \sqcap \abs{x \wedge \neg\phi}_\rho(b) 
      &= \abs{x}_\rho(a) \sqcap \abs{\phi}_\rho(a) \sqcap \abs{x}_\rho(b) \sqcap (\abs{\phi}_\rho(b) \himp \bot) 
    \end{align*}
    If $\rho_{EV}(x) \neq a$, then $\abs{x}_\rho(a) = \bot$, yielding the entire expression above equal to $\bot$, 
    and the same goes for the case when $\rho_{EV}(x) \neq b$. The only case left to consider is thus 
    $a = b = \rho_{EV}(x)$, when we get 
    \begin{align*}
      \abs{x \wedge \phi}_\rho(a) \sqcap \abs{x \wedge \neg\phi}_\rho(b) = \abs{x}_\rho(a) \sqcap \abs{\phi}_\rho(a) \sqcap (\abs{\phi}_\rho(a) \himp \bot) = \bot,
    \end{align*}
    finishing the proof for the $(\operatorname{SINGLETON})$ case, since in general, $h \sqcap (h \himp \bot) = \bot$
    for $h \in \mathbb{H}$.

    \item $(\operatorname{PREFIXPOINT})$ and $(\operatorname{KNASTER\text{-}TARSKI})$: 
      The proof is similar to the classical case, replacing $\mathcal{P}(M)$ with the complete lattice 
      $\mathbb{H}^M$.
      Let $F : \mathbb{H}^M \to \mathbb{H}^M$ be defined by 
      $F(A) = \abs{\phi}_{\rho[X \mapsto A]}$, 
      so that $\abs{\mu X. \phi}_\rho = \operatorname{lfp}(F)$.
      For the $(\operatorname{PREFIXPOINT})$ rule, we need that 
      \begin{align*}
        F(\operatorname{lfp}(F)) = F(\abs{\mu X. \phi}_\rho) = \abs{\phi}_{\rho[X \mapsto \abs{\mu X. \phi}_\rho]} \overset{(*)}{=} \abs{\phi[\mu X. \phi / X]}_\rho \leq \abs{\mu X. \phi}_\rho = \operatorname{lfp}(F),
      \end{align*}
      where $(*)$ is by the substitution lemma, and the inequality actually holds with equality since $\operatorname{lfp}(F)$ is a fixed point of $F$ by Theorem~\ref{thm:kt}.

      For $(\operatorname{KNASTER\text{-}TARSKI})$, suppose that
      $\abs{\phi[\psi/X]}_\rho = \abs{\phi}_{\rho[X \mapsto \abs{\psi}_\rho]} = F(\abs{\psi}_\rho) \leq \abs{\psi}_\rho$.
      We want that $\operatorname{lfp}(F) \leq \abs{\psi}_\rho$.
      From the definition of the least fixed point given by the Knaster-Tarski Theorem \ref{thm:kt},
      it follows immediately, since, by assumption, $\abs{\psi}_\rho$ is a prefixpoint of $F$, 
      that $\operatorname{lfp}(F) \leq \abs{\psi}_\rho$.
      The $\nu$ cases are analogous.

  \end{itemize}

\end{proof}

\section{Conclusions and future work}

To our knowledge, this is the first attempt at an intuitionistic semantics for Matching Logic,
and we believe it represents a starting point for a rich line of future research. 
First of all, while we address the soundness in this paper, we leave open the more difficult question of 
completeness. Even without that, it would be important to validate 
the practical expressivity of the system in terms of the way 
$\ML$ is normally used to specify and to work with different theories, 
by carrying out similar developments in the intuitionistic setting. 
Though we deem our axiomatization to be a natural adaptation of the classical $\ML$ proof system, 
we cannot rule out at this point that one might encounter some missing required proof rules 
when using the system in practice, particularly because the interaction of the purely logical 
and applicative fragments is rather non-standard. 

An immediate first step in this direction would be to study 
in the intuitionistic setting the $\ML$-specific way equality 
is typically introduced as a derived construct, using the theory of definedness. 

Other, more ambitious, potential metatheoretical investigations worth mentioning are
disjunction properties, realizability semantics,
negative translations from $\ML$ to $\iML$, or
the study of Kripke models, which arise in our setting as the particular case where
$\mathbb{H}$ is the algebra of upper sets of a preorder,
and of whether they are already sufficient for completeness.

% \nocite{*}

\newpage \appendix
\section{Proof system for $\iML$}\label{fig:proof-rules}
  
\begin{figure}[h!]
\centering 
{\small
\begin{tabular}{llll}
  $(\operatorname{CONTRACTION})$ & $\phi \vee \phi \to \phi$ & $\phi \to \phi \wedge \phi$ \\[2mm]
  $(\operatorname{WEAKENING})$ & $\phi \to \phi \vee \psi$ & $\phi \wedge \psi \to \phi$ \\[2mm]
  $(\operatorname{PERMUTATION})$ & $\phi \vee \psi \to \psi \vee \phi$ & $\phi \wedge \psi \to \psi \wedge \phi$ \\[2mm]
  $(\operatorname{MODUS\; PONENS})$ & $\dfrac{\phi \quad \phi \to \psi}{\psi}$ \\[4mm]
  $(\operatorname{SYLLOGISM})$ & $\dfrac{\phi \to \psi \quad \psi \to \chi}{\phi \to \chi}$ \\[4mm]
  $(\operatorname{EXPORTATION})$ & $\dfrac{\phi \wedge \psi \to \chi}{\phi \to (\psi \to \chi)}$ \\[4mm]
  $(\operatorname{IMPORTATION})$ & $\dfrac{\phi \to (\psi \to \chi)}{\phi \wedge \psi \to \chi}$ \\[4mm]
  $(\operatorname{EXPANSION})$ & $\dfrac{\phi \to \psi}{\chi \vee \phi \to \chi \vee \psi}$ \\[4mm]
  $(\operatorname{QUANTIFIER})$ & $\phi[y/x] \to \exists x. \phi$ & $\forall x. \phi \to \phi[y/x]$ \\[2mm]
  $(\operatorname{QUANTIFIER\; RULE})$ & $\dfrac{\phi \to \psi}{\exists x. \phi \to \psi}$ & $\dfrac{\psi \to \phi}{\psi \to \forall x. \phi}$, \quad $x \notin FV(\psi)$ \\[4mm]
  $(\operatorname{EXISTENCE})$ & $\exists x. x$ \\[2mm]
  $(\operatorname{SINGLETON})$ & $\neg(C_1[x \wedge \phi] \wedge C_2[x \wedge \neg\phi])$, & where $C_1, C_2$ are application contexts \\[2mm]
  $(\operatorname{PREFIXPOINT})$ & $\phi[\mu X. \phi / X] \to \mu X. \phi$ &
  $\nu X. \phi \to \phi[\nu X. \phi / X]$ \\[2mm]
  $(\operatorname{KNASTER\text{-}TARSKI})$ & $\dfrac{\phi[\psi/X] \to \psi}{\mu X. \phi \to \psi}$
  & $\dfrac{\psi \to \phi[\psi/X]}{\psi \to \nu X. \phi}$ \\[4mm]
  $(\operatorname{SUBSTITUTION})$ & $\dfrac{\phi}{\phi[\psi/X]}$ \\[4mm]
  $(\operatorname{PROPAGATION}_\vee)$ & $(\phi \vee \psi)\chi \to \phi \chi \vee \psi \chi$
  & $\chi(\phi \vee \psi) \to \chi \phi \vee \chi \psi$ \\[2mm]
  $(\operatorname{PROPAGATION}_\exists)$ & $(\exists x. \phi)\chi \to \exists x. \phi \chi$
  & $\chi(\exists x. \phi) \to \exists x. \chi \phi$, \quad where $x \notin FV(\chi)$ \\[2mm]
  $(\operatorname{FRAMING})$ & $\dfrac{\phi \to \psi}{\phi \chi \to \psi \chi}$
  & $\dfrac{\phi \to \psi}{\chi \phi \to \chi \psi}$
\end{tabular}
}
\caption{Proof system for $\iML$.}
\end{figure}

\end{document}